\documentclass[conference]{ieeeconf}
\usepackage[]{times}
\usepackage[top=0.5in, bottom=0.75in, left=0.75in, right=0.75in]{geometry}
\usepackage{epsfig}
\usepackage{subfigure}
\usepackage{amsmath}
\usepackage{amssymb}
\usepackage{graphicx}
\usepackage{subfigure}
\usepackage{multirow}
\usepackage{flushend}
\usepackage{amssymb, amsmath,bm,color}
\usepackage{url}
\usepackage{color}
\newcommand{\bega}{\begin{eqnarray}}
\newcommand{\ega}{\end{eqnarray}}
\newcommand{\bb}{\begin{equation}}
\newcommand{\ee}{\end{equation}}

\newtheorem{defn} {Definition}

\newtheorem{lema}{Lemma}
\newtheorem{cor}{Corollary}

\begin{document}

\title{Worst Configurations (Instantons) for Compressed Sensing over Reals:\\ a Channel Coding Approach}
\author{
\authorblockN{Shashi Kiran Chilappagari}
\authorblockA{Marvell Semiconductor Inc.\\
Santa Clara, CA 95054, USA\\
Email: shashickiran@gmail.com}%
\and
\authorblockN{Michael Chertkov}
\authorblockA{Theoretical Divison
\& CNLS, LANL\\
\& New Mexico Consortium,\\
Los Alamos, NM 87545(4), USA\\
Email: chertkov@lanl.gov}
\and
\authorblockN{Bane Vasic}
\authorblockA{Dept. of ECE\\
University of Arizona\\
Tucson, AZ 85721, USA\\
Email: vasic@ece.arizona.edu}}

\markboth{To appear in proceedings of ISIT 2010}
{Chilappagari, Chertkov \& Vasic: Worst Configurations for Compressed Sensing over Reals}
\maketitle

%\vspace{-0.5in}
\begin{abstract}
We consider the Linear Programming (LP) solution of the Compressed Sensing (CS) problem over reals, also known as the Basis Pursuit (BasP) algorithm. The BasP allows interpretation as a channel-coding problem, and it guarantees error-free reconstruction with a properly chosen measurement matrix and sufficiently sparse error vectors. In this manuscript, we examine how the BasP performs on a given measurement matrix and develop an algorithm to discover the sparsest vectors for which the BasP fails. The resulting algorithm is a generalization of our previous results on finding the most probable error-patterns degrading performance of a finite size Low-Density Parity-Check (LDPC) code in the error-floor regime. The BasP fails when its output is different from the actual error-pattern. We design a CS-Instanton Search Algorithm (ISA) generating a sparse vector, called a CS-instanton, such that the BasP fails on the CS-instanton, while the BasP recovery is successful for any modification of the CS-instanton replacing a nonzero element by zero. We also prove that, given a sufficiently dense random input for the error-vector, the CS-ISA converges to an instanton in a small finite number of steps. The performance of the CS-ISA is illustrated on a randomly generated $120\times 512$ matrix. For this example, the CS-ISA outputs the shortest instanton (error vector) pattern of length $11$.
\end{abstract}

\begin{keywords}
Compressed Sensing, Low-Density Parity-Check Codes, Linear Programming Decoding, Error-floor
\end{keywords}

\section{Introduction}
\label{sec:intro}
\begin{figure}[t]
\centering
\subfigure[]
{
\label{HMatrix}
\includegraphics[width=0.5\textwidth]{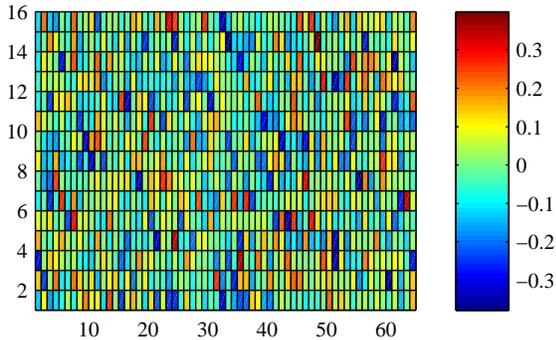}
}
\hspace{-0.3in}
\subfigure[]
{
\label{ErrorVectors}
\includegraphics[width=0.5\textwidth]{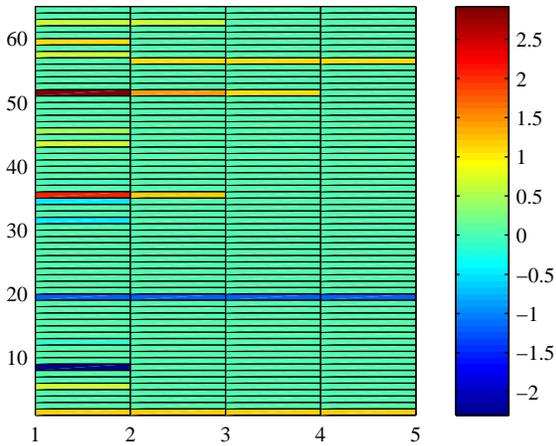}
}
\caption{Fig. (a) shows example of a $(15\times 64)$ measurement matrix, ${F}$, with orthonormal rows. Fig. (b) shows a sequence of error vectors ($e^{(0)}-e^{(3)}$) representing the ISA for the measurement matrix from (a). The first vector (with 15 nonzero elements) was selected at random. It takes four iterations of the ISA (four columns in the Fig. (b))  to reach an instanton, i.e. the error-pattern for which BasP decoding fails, containing only three nonzero elements.
\label{fig:first}}
\end{figure}

\subsection{Background}
Compressed sensing (CS) \cite{06Donoho} is a computational technique to recover a sparse signal from a small set of measurements. Given the measurements, described in terms of the so-called measurement matrix, the ideal CS looks for the sparsest signal, i.e. minimizing the $\ell_0$-norm, consistent with the measurements. This ideal formulation is known to be NP-hard \cite{06Donoho_2,95Natarajan}. A key practical observation of the CS theory came with the Basis Pursuit (BasP) algorithm of Chen, Donoho and Sanders \cite{99CDS} who suggested relaxing the difficult $\ell_0$-norm minimization to  $\ell_1$-norm minimization, which is convex and thus computationally tractable.  The BasP heuristics have shown remarkable performance,  that was theoretically explained in the breakthrough paper of Candes and Tao \cite{05CT}, stating the CS problem as a linear channel coding problem involving recovery of an input real vector from its corrupted image. In \cite{05CT}, it was proved that the $\ell_1$ relaxation guarantees perfect recovery of the input vector for sufficiently sparse error vectors and a properly chosen measurement matrix. The conditions on the measurement matrix were expressed in terms of the so-called Restricted Isometry Property (RIP). It was also shown in \cite{05CT} that a measurement matrix drawn from a proper random ensemble possesses the RIP with high probability.

Estimation of the fraction of tolerable errors for a random measurement matrix became an important follow up question. However, extending the approach of \cite{05CT} to a given finite measurement matrix is not easy, as checking if the matrix has the RIP becomes computationally hard. This aspect of the CS is reminiscent of similar problem in the so-called expander graph based analysis \cite{96SS} of the Low-Density Parity-Check (LDPC) codes \cite{63Gallager}. In fact, expander-based constructions of CS measurement matrices have been investigated in \cite{09JXHC,08Indyk}.

The relation between LDPC codes and CS was further explored in the work of Dimakis and Vontobel \cite{09DV}, who juxtaposed two linear programs: the BasP of CS which can be restated as a Linear Program (LP) and the so-called LP-decoder \cite{05FWK} of a related LDPC code. %(The LP decoder is a sub-optimal decoder whose performance is governed by the parity-check matrix used to define the LP constraints.)
It was shown in \cite{09DV} that a binary matrix which is a good parity check matrix for LP decoding of the corresponding LDPC code is also a good zero-one measurement matrix over reals for BasP in the respective CS problem. The authors of \cite{09DV} also pointed out other directions in which the relationship between BasP and LP-LDPC can be investigated to translate theoretical guarantees from one field to the other. This challenge, of extending descriptive results from the LDPC world into the world of CS reconstruction, has motivated this work.

\subsection{Overview of Results}
In this paper, we pose the following question: {\it Given a measurement matrix good for the $\ell_1$-norm recovery, how does one find the sparsest vectors for which the $\ell_1$-norm minimization fails?} A similar question arises in decoding of a finite LDPC code by Belief Propagation \cite{88Pearl} or other suboptimal decoding algorithms, notably LP decoding, in the so-called error-floor regime \cite{03Richardson,04CCSV,05SCCV}. Roughly speaking, the error floor is an abrupt degradation of the frame error rate performance of a code in the high signal-to-noise ratio (SNR) regime caused by low-weight error patterns not correctable by the sub-optimal decoder. Algorithms to enumerate such rare noise configurations (known as instantons, i.e. instances in the space of noise configurations) are crucial to the design of better decoders as well as codes. Similar statements can be made in the CS setting, where identifying sparsest error vectors for which the BasP fails is important to assess the quality of a given measurement matrix as well as to guide the design of better measurement matrices.

The main contribution of this paper is the formulation of the CS instanton search algorithm (CS-ISA). Our approach to the problem of identifying instantons for CS consists of extending and adapting the ISA developed originally in \cite{08CS,08CCV,09CVSC} for analysis of the error-floor of a given LDPC code, to the problem of the worst configuration recovery in CS. Given a measurement matrix (e.g. the $15 \times 64$ matrix shown in Fig. \ref{HMatrix}), the algorithm starts with a dense error vector for which the BasP fails and iteratively finds error vectors with a smaller $\ell_0$ norm that result in BasP failure. At some point, the algorithm finds an instanton, defined as a sparse error vector not correctable by BasP such that any of its further reductions in the $\ell_0$ norm is BasP correctable. A typical sequence of CS-ISA, described in Section \ref{sec:instanton}, is illustrated in Fig. \ref{ErrorVectors}.

Initializing ISA with an error vector at random allows us to sample the space of instantons for a given measurement matrix efficiently. We describe the quality of the measurement matrix in terms of the distribution of instantons over their sparsity. (See the bar-diagram of Fig.~\ref{fig:bar} for an illustration.) This distribution is the practical tool that we propose to use as a guide for constructing better measurement matrixes.

Finally, we show that the CS-ISA comes with some theoretical guarantees. We prove in Section \ref{sec:proof} that for any sufficiently noisy initialization, the CS-ISA converges in a finite number of steps,  which is significantly smaller than the measurement matrix size.  (Each step of the CS-ISA requires running a single instance of the BasP algorithm.)

After the submission of the initial version of this manuscript, we have been made aware of related approach by Dossal, Peyr\'e and Fadili \cite{10DPF}, who proposed a greedy pursuit algorithm that computes, for a given measurement matrix, sparse vectors for which the BasP fails.

The rest of the manuscript is organized as follows. Section \ref{sec:pre} introduces the problem setting and related terminology. Section \ref{sec:instanton} describes the CS-ISA. Section \ref{sec:example} illustrates the performance of the CS-ISA in sampling the space of instantons of a randomly generated $(120\times512)$ measurement matrix. We present a brief sketch of the theorem, proving the convergence of the algorithm in a small finite number of steps in Section \ref{sec:proof}. Section \ref{sec:summary} is reserved for discussions.

\section{Compressed Sensing Preliminaries: Problem Setting}\label{sec:pre}

In this Section, we discuss the BasP algorithm adopting formulation and terminology from \cite{99CDS,05CT}. (The interested reader is referred to \cite{CS_webpage} for a comprehensive list of references on CS.) The problem setting is as follows. An original real-valued information vector ${f} \in \mathbb{R}^n$ is transformed (coded) into a longer vector  ${A}{f} \in \mathbb{R}^m$, where ${A} \in \mathbb{R}^{m \times n}$ is a full rank matrix known as the generator matrix. The result, transmitted over the CS-channel, is received as ${y}={A}{f}+ {e}$,  where ${e} \in \mathbb{R}^m$ is the unknown error vector, assumed to be sparse. Recovering the error vector $e$ is sufficient to  reconstruct the information vector ${f}$, as knowledge of $y$ along with $e$ gives $Af$ and hence $f$ can be recovered straightforwardly as $A$ is full rank. Now, consider a $p \times m$ matrix $F$ such that $FA=0$. $F$ will be called the measurement matrix. It follows that $\tilde{y}=Fy=Fe$ and the problem of recovering the sparse error vector $e$ is equivalent to the problem of finding the sparsest vector $d$ subject to $Fd=\tilde{y}$. The CS-decoding/reconstruction is successful if $d=e$. The sparsest solution to this problem can be found by solving the following optimization problem (CS-OPT) \cite{05CT}
\begin{eqnarray}
\left.\min_{{d}\in {R}^m} \lVert{d}\rVert_{\ell_0}\right|_{Fd=Fy},
\label{P0}
\end{eqnarray}
where the $\ell_0$ norm of a vector $d$ measures its sparsity, i.e. the number of nonzero entries, $\lVert{d}\rVert_{\ell_0}=|\{i:d_i\neq 0\}|$.

As stated in Eq.~(\ref{P0}) the problem is NP hard \cite{06Donoho_2,95Natarajan} (of exponential complexity in $m$)  and \cite{99CDS} suggested a relaxed (weaker but tractable) version of the CS-decoding,  coined Basis Pursuit (BasP):
\begin{eqnarray}
\left.\min_{{d}\in {R}^m} \lVert{d}\rVert_{\ell_1}\right|_{{F}{d}=\tilde{y}={F}{e}},
\label{P1}
\end{eqnarray}
where $\ell_0$-norm is replaced by the $\ell_1$-norm, $\lVert{d}\rVert_{\ell_1}=\sum_{i=1}^{m}|d_i|$.
Eq.~(\ref{P1}) can also be recast as an LP.

As shown in \cite{05CT}, the BasP is capable of exact reconstruction, i.e. ${d}={e}$,  under the conditions that (1) the measurement matrix has RIP, and (2) $e$ is sufficiently sparse. Formally, \cite{05CT} states that if the RIP constants $\delta_S,\theta_S$ and $\theta_{S,2S}$ satisfy $\delta_S+\theta_S+\theta_{S,2S}<1$, then for any error vector $e$ with $\lVert e \rVert_{\ell_0}\leq S$, $e$ is the unique solution of both Eq.~(\ref{P0}) and Eq.~(\ref{P1}). However, finding the maximum value of $S$ for which the condition $\delta_S+\theta_S+\theta_{S,2S}<1$ holds is, in general, a difficult problem and the most recent state-of-the-art results provide only an estimate for $S/m$ in the asymptotic limit of large samples \cite{05CT}. Moreover, these estimates for the RIP-constants are normally very loose and cannot be used to evaluate the quality of reconstruction of practically important small-to-moderate sized matrices.

On the other hand, brute force search techniques for finding sparse vectors that lead to BasP failures are prohibitively expensive. In fact, BasP with a well-tuned measurement matrix performs on typical instances of the error-vector really well. Thus, using  standard (Monte Carlo) sampling  techniques will typically not deliver a failure. Hence, there is a need to develop smart techniques sampling the space of failures (called instantons) of the given measurement matrix.

\section{Instanton Search Algorithm for BasP}
\label{sec:instanton}
In this Section, we provide a formal description of the CS-ISA. We say that the BasP fails on a vector $e$ if $e\neq d$, where $d$ solves Eq.~(\ref{P1}). We start with the following two definitions.
\begin{defn}
[Instanton]
Let ${e}$ be a $k$-sparse vector (i.e. the number of nonzero entries in ${e}$ is equal to $k$). Consider an error-vector ${e}'$, derived from ${e}$ by replacing one of its nonzero component by zero (thus ${e}'$ is $(k-1)$-sparse). ${e}$ is an instanton if the BasP fails on $e$ while it succeeds on any  ${e}'$ derived from ${e}$.
\label{def:instanton}
\end{defn}

\begin{defn}[Median]
Let ${e}$ be a vector and let $t$ denote the smallest number such that the sum of the $t$ largest entries (in absolute value) of $e$ is at least equal to $\lVert e\rVert_{\ell_1}/2$.  Let $T=\{i_1,i_2,\ldots,i_t\}$ denote indices of the $t$ largest entries of $e$. Then, the median of ${e}$ is the vector $\hat{e}$ with support $T$ and $\hat{e}_i={e_i}$ for $i \in T$.
\label{def:median}
\end{defn}

\textbf{Fact 1:} If ${e}$ is a vector such that $Fe=0$, then BasP fails on $\hat{e}$. To see this, observe that $\lVert e-\hat{e}\rVert_{\ell_1}\leq \lVert \hat{e} \rVert_{\ell_1}$ and $\lVert e-\hat{e}\rVert_{\ell_0}\geq \lVert \hat{e} \rVert_{\ell_0}$.

Now we are ready to describe the {\bf Instanton Search Algorithm}:
\begin{itemize}
\item \underline{Initialization ($l=0$) step:} Initialize the algorithm to a vector $e^{(0)}$  of length $m$ with a sufficient number of errors such that BasP fails on $e^{(0)}$,  i.e. applied to $e^{(0)}$ BasP produces another vector $\bar{e}^{(0)} \neq e^{(0)}$.
\item \underline{$l \geq 1$ step:} Consider the vector $\tilde{e}^{(l-1)}=e^{(l-1)}-\bar{e}^{(l-1)}$, where $\bar{e}^{(l-1)}$ denotes the output of BasP acting on $e^{(l-1)}$. Let $\hat{e}^{(l-1)}$ denote the median of $\tilde{e}^{(l-1)}$. Only two cases arise (see Lemma \ref{lemma1}): \\ (i)  If $\lVert\hat{e}^{(l-1)}\rVert_{\ell_0} < \lVert e^{(l-1)}\rVert_{\ell_0}$, then $e^{(l)}=\hat{e}^{(l-1)}$ is the $l$-th step output/($l+1$)-th step input. (Note that by Fact 1, BasP fails on $e^{(l)}$)\\
(ii) If $\lVert\hat{e}^{(l-1)}\rVert_{\ell_0} = \lVert e^{(l-1)}\rVert_{\ell_0}$, define $L=\{i_1,i_2,\ldots,i_{k_l}\}$ as the support of $\hat{e}^{(l-1)}$. Let $L_{i_t}=L \backslash i_t$ for some $i_t \in L$. Let $r^{i_t}$ be a vector such that $r^{i_t}_j=\hat{e}^{(l-1)}_j, j \in L_{i_t}$ and contains zero elements elsewhere. Apply BasP to all $r^{i_t}$ and denote the $i_t$ output by $\bar{r}^{i_t}$. If $\bar{r}^{i_t}=r^{i_t}$ for all $i_t$, then $\hat{e}^{(l-1)}$ is the desired instanton and the algorithm halts. Else, $e^{(l)}$ is set to $r^{i_t}$ for some $i_t$ for which $\bar{r}^{i_t} \neq r^{i_t}$.
\end{itemize}

Starting with some random non-sparse initialization of the error vector for which BasP fails, the algorithm aims to get iteratively as close as possible to the zero error vector while keeping the failed status for BasP when it acts on the current vector. In the ISA (which should in fact, for the sake of accuracy, be called ISA for the $l_0$-norm channel), this aim is achieved by alternating the BasP- and the median- steps.
The logic behind the alternation is obvious: when BasP fails we apply the median-step to reduce the size of $\tilde{e}$, defined as the difference between the original $e$ and $\bar{e}$.

We will prove in Section \ref{sec:proof} that the ISA algorithm defined above outputs an instanton in a small number of steps. However, prior to that we illustrate the algorithm performance on an example in the next Section.

\section{Example of the Error-Surface Exploration}
\label{sec:example}
In this Section, we first consider a $15 \times 64 $ measurement matrix, used to show the performance of the ISA in Fig.~\ref{fig:first}. This toy example, described in details in the next paragraph, while not useful in practice, allows a reasonable visualization and thus serves us mainly for illustration purposes. Afterwards, and focused on a tour-de-force demonstration, we discuss a more realistic $120\times 512$ example.

The $15 \times 64$ measurement matrix, shown in Fig.~\ref{HMatrix}, is built by first creating a $15 \times 64$ matrix $H$ with i.i.d. Gaussian entries and then finding the orthonormal basis of the space spanned by the columns of $H$. The dynamics of the ISA is illustrated in Fig.~\ref{ErrorVectors}. The error vector initialization is random and it happens to be with support of size $15$ for the illustrative example. For this example locations of the nonzero components of the initialization are, $[1~5~8~12~19~31~34~35~37~43~45~51~57~59~62]$, and their respective values are, $[1.1738~0.6554~-~2.2990~-~0.1783~-~1.1907~-~0.4254\\~-~0.4768~2.0385 ~-~0.0695~0.6997~0.4137~2.9185~0.6545\\~1.1149~0.6789]$. BasP fails on this error pattern producing the output $\bar{e}^{0}$. In step $1$, we compute the median of $\hat{e}^{0}$ resulting in a vector with support $[1~19~35~51~56~62]$ and components $[1.1738~-1.1907~1.1590~1.3883~1.0888~0.6789]$. Next, we set $e^{(1)}=\hat{e}^{(0)}$. In step $2$, we compute the median of $\hat{e}^{(1)}$ resulting in a vector  with support $[1~19~51~56]$ and components $[1.1738~-1.1907~1.0843~1.0888]$. Hence, we set $e^{(2)}=\hat{e}^{(1)}$. In step $3$, we compute $\hat{e}^{(2)}$ yielding the vector with support $[1~19~56]$ and components $[1.1738~-1.1907~1.0888]$, thus resulting in $e^{(3)}=\hat{e}^{(2)}$. In step $4$, the median of $\tilde{e}^{(3)}$ generates vector with support $[1~19~56]$ and components $[1.1738~-1.1907~1.0138]$. Since,  $\lVert\hat{e}^{(3)}\rVert_{\ell_0} = \lVert e^{(3)}\rVert_{\ell_0}$, we consider all the error vectors derived from $\hat{e}^{(3)}$ by replacing one (of the three) nonzero components by zero. BasP applied to these newly derived vectors decodes correctly,  and thus $\hat{e}^{(3)}$ is declared an instanton of the measurement matrix.

In order to illustrate the effectiveness of the ISA in sampling the space of instantons, we consider a random sample of $(120 \times 512)$ measurement matrix with orthogonal rows. We run the ISA $5000$ times, generating initialization for the error-vector always with $40$ initial errors and record the length of the instanton found for each trial. ($40$ is a sufficiently large integer chosen to guarantee that the majority of random error-vector initializations leads to BasP failure and to an instanton configuration consequently, while only a small portion of initializations of the $\ell_0$-norm $40$ are decoded correctly by BasP,  and these rare and not interesting instances are simply ignored. Note that lowering the initialization sparsity will lead to wasting a lot of initializations as these would be typically decoded by BasP without an error.)  The resulting distribution of the instanton lengths is shown in Fig.~\ref{fig:bar}. The smallest length instanton discovered by the ISA for this example is $11$.
Note that it is still possible that there exist instantons of length less than $11$.

The ISA was implemented in MATLAB and the BasP step in the algorithm was solved using the ``l1-magic" package from \cite{l1magic}. For the aforementioned measurement matrix of size $(120\times 512)$ and the initial error-vector of the $\ell_0$-norm $40$, the average time of the BasP and the ISA runs were $0.2$ seconds and $4.2$ seconds respectively on a (laptop) Intel core duo 2GHz processor.

\begin{figure}
\includegraphics[width=0.5\textwidth]{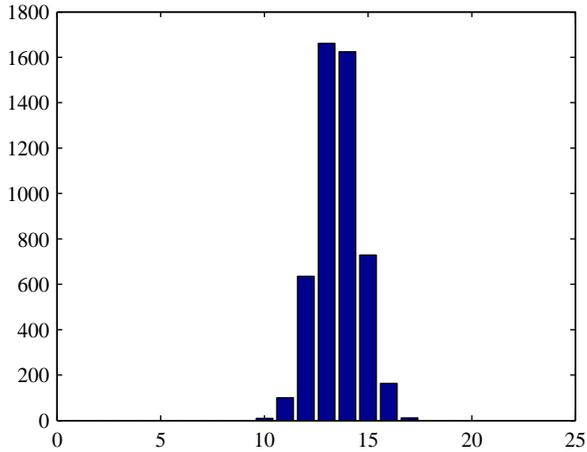}
\caption{Bar-diagram showing number of instance per recorded instanton length found in the result of $5000$ trial runs of the ISA on a random measurement matrix of size $(120\times 512)$ and with random initializations of the sparsity ($\ell_0$-norm) $40$. \label{fig:bar}
}
\end{figure}

\section{Proof of the CS-IS Algorithm Correctness}

In this Section, we establish correctness of the CS-ISA and prove that the ISA outputs an instanton in a finite number of steps.
\label{sec:proof}
\begin{lema}\label{lemma1} Let $e$ be a vector and let $\bar{e}$ denote the output of BasP on $e$. Further, let $\tilde{e}=e-\bar{e}$ and denote the median of $\tilde{e}$ by $\hat{e}$. Then,
\[
\lVert\hat{e}\rVert_{\ell_0}\leq\lVert e\rVert_{\ell_0}
\]
\end{lema}
\begin{proof}
Let $S_1$ denote the support of $e$ and $J$ denote the support of $\bar{e}$. Since $\bar{e}$ is the BasP output for $e$, we have
\[
\lVert e\rVert_{\ell_1}\geq \lVert\bar{e}\rVert_{\ell_1}
\]
Or stating it differently,
\begin{eqnarray}
\sum_{i \in S_1}|e_i|&\geq& \sum_{j \in J}|\bar{e}_j|\\
\Rightarrow \sum_{i \in S_1}(|e_i|-|\bar{e}_i|)&\geq&\sum_{j \in J \setminus S_1}|\bar{e}_j|
\end{eqnarray}
Now consider, $\sum_{i \in S_1}|\tilde{e}_i|$. We have
\begin{eqnarray}
\sum_{i \in S_1}|\tilde{e}_i|&=&\sum_{i \in S_1}|e_i-\bar{e}_i| \geq \sum_{i \in S_1}(|e_i|-|\bar{e}_i|) \geq \sum_{j \in J \backslash S_1}|\bar{e}_j| \nonumber \\
%&\geq&\sum_{i \in S_1}(|e_i|-|\bar{e}_i|) \\
%&\geq&\sum_{j \in J \backslash S_1}|\bar{e}_j|\\
 &=&\sum_{j \in J \backslash S_1}|e_j-\bar{e}_j|=\sum_{j \in J \backslash S_1 }|\tilde{e}_j|  \label{eq1_5}
\end{eqnarray}
Eq. (\ref{eq1_5}) implies that the sum of entries of $\tilde{e}$ in $S_1$ is at least equal to the sum of entries of $\tilde{e}$ over the remaining entries. This implies that the $\ell_0$ norm of the median of $\tilde{e}$ cannot exceed $|S_1|=\lVert e\rVert_{\ell_0}$.
\end{proof}
\begin{cor}
The ISA converges to an instanton in a finite number of steps.
\end{cor}
\begin{proof}
At every step $l$ of the ISA, we have $\lVert e^{(l)}\rVert_{\ell_0}<\lVert e^{(l-1)}\rVert_{\ell_0}$. By construction, the output of the ISA is an instanton.
\end{proof}

\section{Discussion and Future Work}
\label{sec:summary}
We envision further development of the Instanton Search technique along the following lines.

\noindent$\bullet$ Sampling of the instanton space provided by the ISA is not uniform.  An important future task is to design a uniform-sampling modification of the ISA.

\noindent$\bullet$ Studying how ISA complexity scales with the measurement matrix size constitutes another important problem. ISA outputs an instanton in a linear (in the number of columns of measurement matrix) number of steps in the worst case. However, it is still not clear how many instances of the ISA should be executed till the smallest instanton is found. While the answer clearly depends on the particular measurement matrix, to answer this question for an average case would also be of interest.

\noindent$\bullet$ If the error-vector is almost sparse, or alternatively if the measurements are noisy, the problem of exact CS reconstruction is replaced by an approximate reconstruction. One possible modification of the BasP algorithm, the so-called Lasso algorithm \cite{94Tib}, consists in  adding to the $\ell_1$-norm of the signal (the original BasP) a part linear in the $\ell_2$-norm of the noise. With a proper (soft) definition of failure, one should be able to extend the ISA approach to finding the worst (highest probability) configurations leading to Lasso's failures.

 \noindent$\bullet$ The ISA technique allows adaptation to the problem of Matrix Completion via the computationally tractable minimization of the matrix nuclear norm (replacing exact but not tractable minimization of the matrix rank) \cite{08CR,09CT}. In this setting, a sparse and properly conditioned matrix (which one aims to reconstruct) is fixed and the question becomes exploring the set of measurements required for the nuclear-norm minimization to succeed. The desired modification of the ISA should be capable of sampling the most dangerous configurations of measurement,  for example defined as a set of measurements such that their nuclear-norm based decoding leads to a failure,  while addition of any single non-zero measurements results in a successful decoding.

 \noindent$\bullet$ We can also adapt the ISA algorithm to the sequential compressed sensing formulation of Maloutov \textit{et al} \cite{09MSW}.  In fact,  the modification is rather straightforward as it only replaces BasP,  as a sub-step of the ISA, by its sequential implementation (also imitating the order in which the measurements are received). A similar modification can be implemented for any other LP-based modification of BasP,  such as the Subspace Pursuit algorithm discussed in \cite{08DM}. Note that this streamlining of BasP is somehow similar to the adaptive realization of the LP decoding of LDPC codes suggested in \cite{08TS}.

\noindent$\bullet$ However successful the BasP algorithm is, it is still suboptimal with respect to the exact $\ell_0$-norm minimization. It will thus be important to study the gap between exact and suboptimal decodings, in particular constructing a sequence of convex optimizations improving the performance of BasP gradually. This sequence may be designed along the lines of discussion in \cite{09DV}, illustrated on example of the $0-1$ measurement matrix and error-vector that the LP-LDPC is stronger than the respective BasP. The ISA-approach can be easily adapted to such convex improvements over BasP.

\noindent$\bullet$ We envision that the most interesting (but also the most challenging) application of this Instanton Search approach will be in designing a good measurement matrix, in the spirit of how one may think about using the instanton search for selecting a good LDPC-parity check matrix (say picked from given random ensemble optimized with respect to its water-fall behavior) with the lowest error-floor.

\section{Acknowledgments}
The work was carried out under the auspices of the NNSA of the U.S. DOE at LANL under Contract No. DE-AC52-06NA25396. The work by S. K. Chilappagari was performed when he was with the Department of Electrical and Computer Engineering, University of Arizona, Tucson, AZ. The work of M. Chertkov was also funded by UC-LANL Award No. 09-LR-06-118620-SIEP and by NSF CCF-0829945 (via NMC). The work of B. Vasic and S. K. Chilappagari was funded by the NSF under Grants CCF-0634969, IHCS-0725405 and CCF-0830245.


\end{document}